\documentclass[reqno,11pt]{amsart}

\usepackage{amsmath,amssymb,amsthm}    
\usepackage[babel]{microtype}
\usepackage[dvipsnames]{xcolor}
\usepackage{bookmark}
\usepackage{verbatim}
\usepackage{subcaption}
\usepackage[export]{adjustbox}
\usepackage[noend]{algpseudocode}
\usepackage{algorithm}

\usepackage{datetime}
\usepackage{lineno}
\newcommand*\patchAmsMathEnvironmentForLineno[1]{%
\expandafter\let\csname old#1\expandafter\endcsname\csname #1\endcsname
\expandafter\let\csname oldend#1\expandafter\endcsname\csname end#1\endcsname
\renewenvironment{#1}%
{\linenomath\csname old#1\endcsname}%
{\csname oldend#1\endcsname\endlinenomath}}%
\newcommand*\patchBothAmsMathEnvironmentsForLineno[1]{%
\patchAmsMathEnvironmentForLineno{#1}%
\patchAmsMathEnvironmentForLineno{#1*}}%
\AtBeginDocument{%
\patchBothAmsMathEnvironmentsForLineno{equation}%
\patchBothAmsMathEnvironmentsForLineno{align}%
\patchBothAmsMathEnvironmentsForLineno{flalign}%
\patchBothAmsMathEnvironmentsForLineno{alignat}%
\patchBothAmsMathEnvironmentsForLineno{gather}%
\patchBothAmsMathEnvironmentsForLineno{multline}%
}

\usepackage{geometry}
\usepackage{pgf}
\usepackage{tikz}
\usepackage{pgfplots}
\usepackage{xcolor}

\usetikzlibrary{shapes,arrows}
\usetikzlibrary{calc}
\usetikzlibrary{external}
\usetikzlibrary{positioning} 
\tikzstyle{seed}=[-latex,red, thick]
\tikzstyle{ray}=[green!50!black,dotted,thick,-latex]
\usepackage{tkz-graph}
\usepackage{enumerate}   

\usepackage[T1]{fontenc}
\usepackage[english]{babel}

\newtheorem{theorem}             {Theorem}[section]
\newtheorem{lemma}     	[theorem] {Lemma}

\newcommand{\NP}{\mbox{NP}}
\newcommand{\PP}{\mbox{P}}
\newcommand{\IS}{{\sc Independent Set}}
\newcommand{\DM}{{\sc 3D-Matching}}
\newcommand{\eps}{\epsilon}
\newcommand{\Dag}{dag}
\newcommand{\opt}{\mathrm{opt}}
\newcommand{\Oh}{\mathrm{O}}
\newcommand{\maxleaves}{{\sc Maximum Leaf Spanning Arborescence}}
\newcommand{\maxsnp}{MaxSNP}
\newcommand{\calS}{\mathcal{S}}

\begin{document}

\title{Leafy Spanning Arborescences in {DAG}s}

\thanks{This research was conducted while the authors were attending the 3rd~WoPOCA: ``Workshop Paulista em Otimiza\c c\~ao, Combinat\'oria e Algoritmos''.  An extended abstract of this work is to appear in the proceedings of the 14th Latin American Theoretical Informatics Symposium (LATIN)~\cite{FernandesL2020}.
C. G. Fernandes was partially supported by CNPq (Proc.~308116/2016-0 and~423833/2018-9).}

\author{Cristina G. Fernandes and Carla N. Lintzmayer}

\shortdate
\yyyymmdddate
\settimeformat{ampmtime}
\date{\today, \currenttime}

\address{Institute of Mathematics and Statistics. University of S{\~a}o Paulo. S{\~a}o Paulo, Brazil}
\email{\href{mailto:cris@ime.usp.br}{\texttt{cris@ime.usp.br}}}

\address{Center for Mathematics. Computing and Cognition. Federal University of ABC. Santo Andr{\'e}, S{\~a}o Paulo, Brazil}
\email{\href{mailto:carla.negri@ufabc.edu.br}{\texttt{carla.negri@ufabc.edu.br}}}

\begin{abstract}
    Broadcasting in a computer network is a method of transferring a message to all recipients simultaneously.
    It is common in this situation to use a tree with many leaves to perform the broadcast, as internal nodes have to forward the messages received, while leaves are only receptors.
    We consider the subjacent problem of, given a directed graph~$D$, finding a spanning arborescence of~$D$, if one exists, with the maximum number of leaves.  
    In this paper, we concentrate on the class of rooted directed acyclic graphs, for which the problem is known to be \maxsnp-hard. 
    A 2-approximation was previously known for this problem on this class of directed graphs.
    We improve on this result, presenting a $\frac32$-approximation.
    We also adapt a result for the undirected case and derive an inapproximability result for the vertex-weighted version of \maxleaves\ on rooted directed acyclic graphs.
\end{abstract}

\maketitle

\section{Introduction}

The problem of, given a connected undirected graph, finding a spanning tree with the maximum number of leaves is well known in the literature, appearing as one of the NP-hard problems in the classic book by Garey and Johnson~\cite{GareyJ1979}.  
With many applications in network design problems, the best known result for it is a long standing 2-approximation proposed by Solis-Oba~\cite{SolisOba1998,SolisObaBL2017}.
In the literature, a directed version of this problem has also been considered.

For network broadcast, one looks for a directed spanning tree rooted at a source node, in which all arcs are directed away from the source.
Broadcast trees with many leaves are preferable in this situation~\cite{JuttnerM2005,PopeS2015}.
Internal nodes have to forward the messages received, while leaves are only receptors.
Also, in some applications, it is interesting to build a more robust backbone tree, and possibly less expensive links to reach the endpoint clients.
The cost of such a backbone tree is usually related to its number of arcs.
By maximizing the number of leaves in a rooted directed spanning tree, we are minimizing the number of arcs in the tree obtained from removing the arcs incident to the leaves, which can be seen as a backbone tree for the network.
To define the directed version of the problem precisely, we introduce some notation.

Let $D$ be a directed graph.
A vertex $r$ in $D$ is a \emph{root} if there is a directed path in $D$ from $r$ to every vertex in $D$. 
If $r$ is a root in $D$, then we say $D$ is \emph{$r$-rooted}, or simply \emph{rooted}.
We say $D$ is \emph{acyclic} if there is no directed cycle in~$D$.
A directed acyclic graph is called a \emph{dag}, for short.
Note that any rooted dag has only one root.
An \emph{arborescence} is an $r$-rooted dag $T$ for which there is a unique directed path from $r$ to every vertex in $T$. 
The \emph{out-degree} of a vertex in a directed graph is the number of arcs that start in that vertex, while the \emph{in-degree} of a vertex is the number of arcs that end in that vertex.
A vertex of out-degree~$0$ in an arborescence is called a \emph{leaf}.  

The \maxleaves\ is the problem of, given a rooted directed graph $D$, finding a spanning arborescence of $D$ with the maximum number of leaves.  
Let $\opt(D)$ denote the number of leaves in such an arborescence.

Given an undirected graph $G$, one can consider the digraph $D$ obtained by substituting each edge by two arcs, one in each direction.  
With this construction, it is easy to deduce that the \maxleaves\ is NP-hard, as its undirected version.  
Alon et al.~\cite{AlonFGKS2009} showed that the \maxleaves\ remains NP-hard on dags.  
They were in fact investigating whether the \maxleaves\ is fixed parameter tractable~\cite{CyganFKLMPPS2016}, and they gave a positive answer for strongly connected digraphs, as well as for dags.  
Binkele-Raible et al.~\cite{BinkeleRaibleFFLSV2012} provided a cubic size kernel for the \maxleaves, 
and Daligault and Thomass\'{e}~\cite{DaligaultT2009} improved on this result, providing a quadratic size kernel.  
It is worth mentioning that a linear size kernel is known for the undirected version of the problem.  

As a byproduct, Daligault and Thomass\'{e}~\cite{DaligaultT2009} derived a 92-approximation for the \maxleaves\ in general rooted directed graphs.  
This turns into a 24-approximation when the digraph has no digon (directed cycle of length two).  
More recently, Schwartges, Spoerhase, and Wolff~\cite{SchwartgesSW2012} described a 2-approximation for the case in which the digraph is acyclic, and proved that this restricted version of the \maxleaves\ is \maxsnp-hard.
Their algorithm is inspired on a greedy 3-approximation by Lu and Ravi~\cite{LuR1998} for the undirected version of the problem.  

Sections~\ref{sec:algo} and \ref{sec:ratio} present a $\frac32$-approximation algorithm for the \maxleaves\ on rooted dags. 
Our algorithm is somehow inspired on Solis-Oba's algorithm, in the sense that it prioritizes certain expansion rules.
However, there is a key difference: in one of the rules, the number of expansions can be optimized. 
Section~\ref{sec:3dmatching} explores the relation of our algorithm with matchings.
Section~\ref{sec:inapproximability} shows an inapproximability result for the vertex-weighted version of \maxleaves\ on rooted dags.
Some open problems are presented in Section~\ref{sec:remarks}.

\section{The algorithm}
\label{sec:algo}

A \emph{branching} is a forest of arborescences.
A vertex that is not a leaf in a branching is called \emph{internal}. 
For a positive integer~$t$, a~\emph{$t$-branching} is a branching all of whose internal vertices have out-degree at least $t$.
See Figure~\ref{fig:branching}.

\begin{figure}[htb]
\centering
\begin{subfigure}{0.46\textwidth}
  \includegraphics[width=1\textwidth,left]{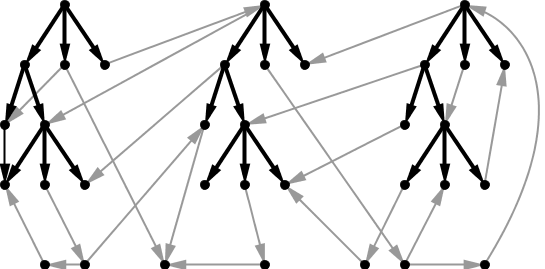}
\caption{A 2-branching.}\label{figa:branching}
\end{subfigure}
\hspace{5mm}
\begin{subfigure}{0.46\textwidth}
  \includegraphics[width=1\textwidth,right]{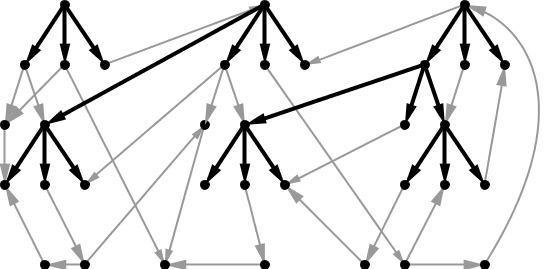}
\caption{A 3-branching.}\label{figb:branching}
\end{subfigure}
\caption{The bold arcs show two different branchings in a rooted dag.}
\label{fig:branching}
\end{figure}

For a directed graph $D$, we denote by $V(D)$ and $A(D)$ the set of vertices and arcs of $D$ respectively.
For a vertex $v$ in~$V(D)$, we denote by $d_D^+(v)$ its out-degree in~$D$ and by $d_D^-(v)$ its in-degree in~$D$.  
The \emph{out-neighbors} of $v$ are the extreme vertices of arcs that start at $v$.
We say a spanning $t$-branching is maximal if, for any vertex of out-degree~$0$, its set of out-neighbors with in-degree~$0$ contains less than~$t$ vertices.
The spanning branchings in Figure~\ref{fig:branching} are maximal.

Algorithm~\ref{alg:2} presents \Call{GreedyExpand}{$D$, $t$, $F$}, the heart of our approximation.
Given a rooted \Dag\ $D$, a positive integer $t$, and a spanning $(t+1)$-branching~$F$ of~$D$, it returns a maximal spanning $t$-branching of~$D$ containing~$F$.

\begin{algorithm}
\begin{algorithmic}
\Require{{\small rooted \Dag\ $D$, a positive integer $t$, and a spanning $(t{+}1)$-branching~$F$ of~$D$}}
\Ensure{{\small a maximal spanning $t$-branching of $D$ containing~$F$}}

\State $F' \gets F$
\For{each $v \in V(D)$ such that $d_{F'}^+(v)=0$}
    \State $U_v \gets \{vu \in A(D) : d_{F'}^-(u)=0\}$
    \If{$|U_v| \geq t$}
        \State $F' \gets F' + U_v$
    \EndIf
\EndFor
\State \Return $F'$
\end{algorithmic}
\caption{\textsc{GreedyExpand}($D$, $t$, $F$)}
\label{alg:2}
\end{algorithm}

Let us argue that the call \Call{GreedyExpand}{$D$, $t$, $F$} produces a maximal $t$-branching.  
Indeed, the returned $F'$ is spanning because~$F'$ contains the spanning branching~$F$.
Also, all internal vertices of~$F'$ have in-degree at most one and out-degree at least $t$.
So $F'$ is a $t$-branching and is clearly maximal.

The branchings in Figure~\ref{fig:branching} are possible outputs of the calls \Call{GreedyExpand}{$D$, $2$, $F$} and \Call{GreedyExpand}{$D$, $3$, $F$}, respectively, when $D$ is the depicted dag and $F$ is the spanning branching of $D$ with no arcs.

We observe that the \Call{GreedyExpand}{} is an extension of the \Call{Expansion}{} algorithm by Schwartges, Spoerhase, and Wolff~\cite{SchwartgesSW2012}. 
Particularly, if $F$ is the spanning branching of $D$ with no arcs, then \Call{GreedyExpand}{$D$, $2$, $F$} behaves as \Call{Expansion}{$D$} on any rooted dag~$D$.  

Algorithm~\ref{alg:maxleaves} shows our approximation for the \maxleaves\ on rooted \Dag s, named \Call{MaxLeaves}{}. 
It uses twice the previously presented \Call{GreedyExpand}{}.
Algorithm \Call{MaxLeaves}{} also uses an algorithm \Call{MaxExpand}{$D$,~$F$} that receives a rooted \Dag\ $D$ and a maximal spanning $3$-branching~$F$ of~$D$, and returns a maximum spanning $2$-branching of~$D$ containing~$F$.
Algorithm \Call{MaxExpand}{$D$,~$F$} will be described after \Call{MaxLeaves}{}. 

\begin{algorithm}
\begin{algorithmic}
\Require{rooted acyclic directed graph $D$}
\Ensure{spanning arborescence with at least $\frac32\,\opt(D)$ leaves}

\State let $F_0$ be the spanning branching with no arcs
\State $F_1 \gets$ \Call{GreedyExpand}{$D$, $3$, $F_0$}
\State $F_2 \gets$ \Call{MaxExpand}{$D$, $F_1$}
\State $T \gets$ \Call{GreedyExpand}{$D$, $1$, $F_2$}
\State \Return $T$
\end{algorithmic}
\caption{\textsc{MaxLeaves}($D$)}
\label{alg:maxleaves}
\end{algorithm}

The call \Call{GreedyExpand}{$D$, $1$, $F$} returns a maximal 1-branching of the rooted dag $D$ containing $F$, that is, a spanning arborescence of~$D$ containing~$F$.
So \Call{MaxLeaves}{$D$} indeed produces a spanning arborescence of~$D$. 
See Figure~\ref{fig:maxleaves}.
In the next section, we will prove that algorithm \Call{MaxLeaves}{} is a~$\frac32$-approximation for the \maxleaves\ on rooted dags.

\begin{figure}[htb]
\centering
\begin{subfigure}{0.46\textwidth}
  \includegraphics[width=1\textwidth]{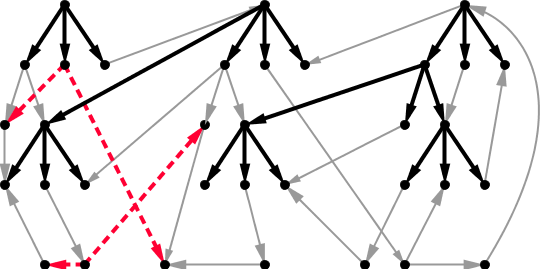}
  \caption{Branching $F_2$.}\label{figa:phase2}
\end{subfigure}
\hspace{5mm}
\begin{subfigure}{0.46\textwidth}
  \includegraphics[width=1\textwidth]{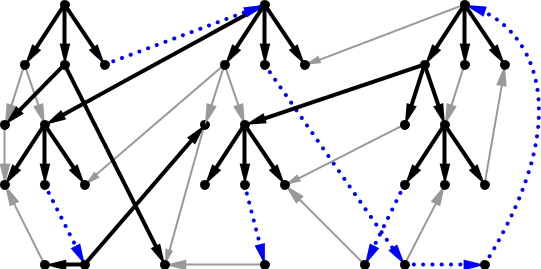}
\caption{Spanning arborescence $T$.}\label{figb:phase3}
\end{subfigure}
\caption{In the left, the bold arcs represent a possible maximal 3-branching $F_1$ and the dashed arcs were added to obtain $F_2$.  In the right, the bold arcs represent $F_2$ and the dotted arcs were added to obtain $T$.}
\label{fig:maxleaves}
\end{figure}

The \Call{MaxExpand}{$D$, $F$} procedure, presented in Algorithm~\ref{alg:3}, is an optimized version of \Call{GreedyExpand}{$D$, $2$, $F$}.
It uses an algorithm \mbox{\Call{MaximumMatching}{}} that receives an undirected multigraph~$G$ and returns a maximum matching in~$G$.
Polynomial-time algorithms for this are known in the literature~\cite{Edmonds1965}.

\newcommand{\Candidates}{\mathit{Candidates}}
\begin{algorithm}
\begin{algorithmic}
\Require{rooted \Dag\ $D$ and a maximal spanning $3$-branching~$F$ of $D$}
\Ensure{a maximum spanning $2$-branching of $D$ containing~$F$}

\For{each $v \in V(D)$ such that $d_F^+(v)=0$}
    \State $U_v \gets \{vu \in A(D) : d_F^-(u)=0\}$
\EndFor
\State $\Candidates \gets \{v \in V(D) : d_F^+(v)=0 \mbox{ and } |U_v|=2\}$
\State $V' \gets \{u \in V(D) : d_F^-(u)=0\}$
\State $E' \gets \{e_v = uw : v \in \Candidates \mbox{ and } U_v = \{vu,vw\}\}$
\State let $G$ be the undirected multigraph $(V',E')$
\State $M \gets$ \Call{MaximumMatching}{$G$}
\State $F' \gets F$
\For{each $e_v \in M$}
    \State $F' \gets F' + U_v$
\EndFor
\State \Return $F'$%
\end{algorithmic}
\caption{\textsc{MaxExpand}($D$, $F$)}
\label{alg:3}
\end{algorithm}

The call \Call{MaxExpand}{$D$, $F$} produces a maximum spanning $2$-branching of~$D$ containing~$F$.
It does this by constructing an undirected multigraph~$G$ whose vertices are vertices of in-degree~$0$ in~$F$ and an edge~$uw$ exists in~$G$ if~$u$ and~$w$ are the only out-neighbors of in-degree~$0$ in~$F$ of some vertex~$v$ of out-degree~$0$ in~$F$.
Thus, edge~$uw$ of~$G$ represents an expansion that can be performed on vertex~$v$ of~$D$.
The fact that more than one vertex of out-degree~$0$ in~$F$ may have vertices~$u$ and~$w$ of in-degree~$0$ in~$F$ as their out-neighbors shows the need for a multigraph.
Independent edges in this undirected multigraph correspond to compatible expansions, so a maximum matching gives the maximum number of expansions that can be performed in~$D$. 
See Figure~\ref{fig:matching}.

Indeed, note that, for the returned $F'$ to be a branching, the edges $e_v$ corresponding to expanded vertices~$v$ must form a matching in the multigraph~$G$. 
Otherwise, there would be vertices with in-degree greater than one. 
As $F$ is a maximal $3$-branching and~$D$ is acyclic, the returned $F'$ is also a branching, and therefore a maximum $2$-branching containing~$F$.
See Figure~\ref{fig:matching2}.

We observe that, in the dag shown in Figure~\ref{fig:matching}, our algorithm produces the best arborescence possible, with roughly half of the vertices of the dag as leaves.
Meanwhile, the algorithm due to Schwartges, Spoerhase, and Wolff~\cite{SchwartgesSW2012} could have produced an arborescence with only one forth of the vertices as leaves.

\begin{figure}[h]
\centering
\begin{subfigure}{\textwidth}
  \includegraphics[width=.8\textwidth,center]{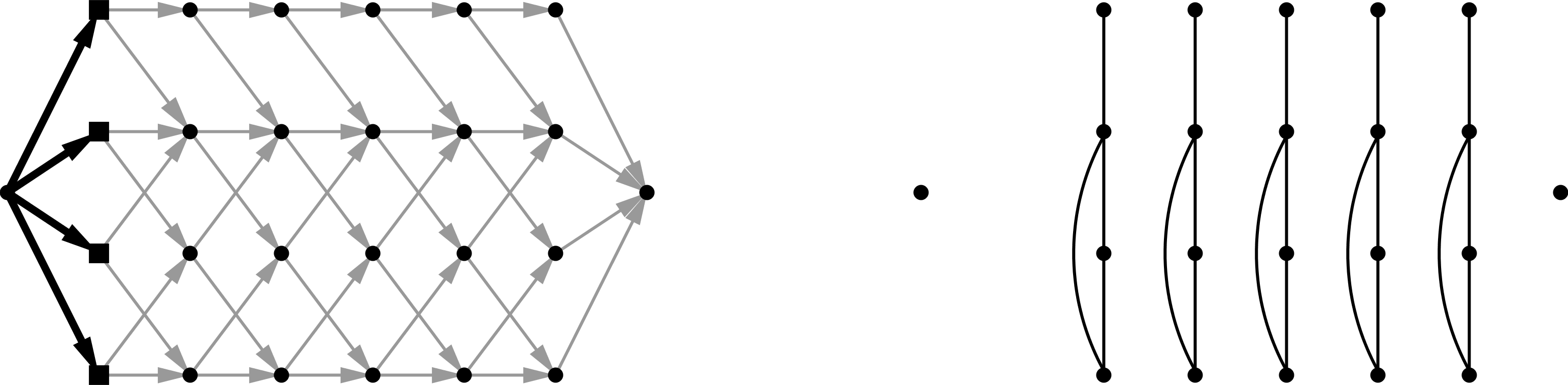}
  \caption{The bold arcs incident to the root of the dag are a maximal 3-branching. 
    The round vertices form the set $V$, and the corresponding multigraph $G$ is in the right.}
  \label{fig:matching}
\end{subfigure}
\begin{subfigure}{\textwidth}
  \bigskip
  \includegraphics[width=.8\textwidth,center]{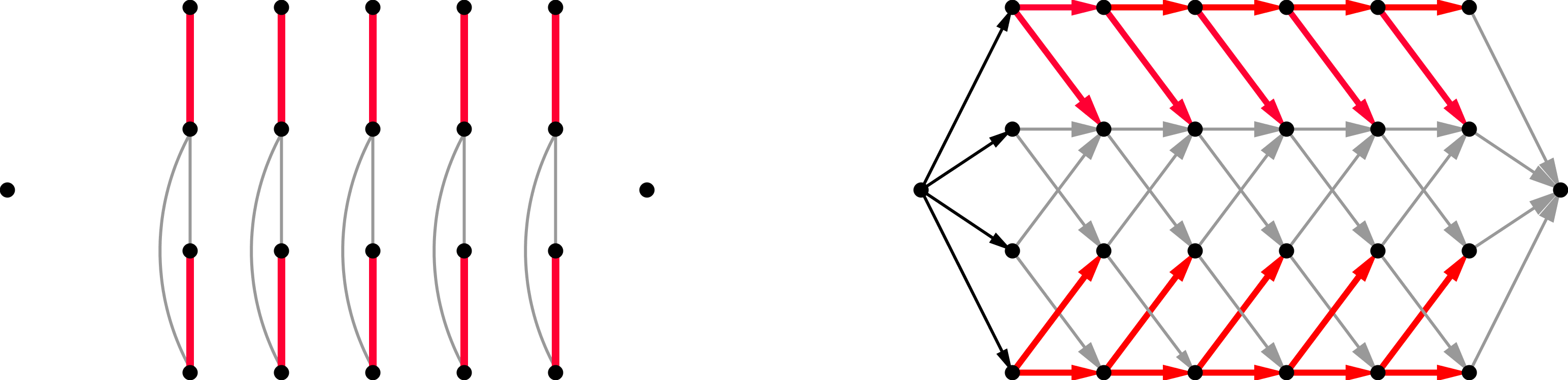}
  \caption{In the left, a maximum matching in red and bold. In the right, the corresponding expansions in red and bold.}
  \label{fig:matching2}
\end{subfigure}
\caption{Example of an execution of \textsc{MaxExpand}.}
\label{fig:maxexpand}
\end{figure}

\section{Approximation ratio}
\label{sec:ratio}

Let $F_1$, $F_2$, and $T$ be the branchings produced in the call \Call{MaxLeaves}{$D$}.
For~$i=1,2$, let~$k_i$ be the number of non-trivial components of $F_i$ and $N_i$ be the number of vertices in such components.
We denote by $\ell(F)$ the number of leaves in any branching $F$.

For example, if $D$ is the dag depicted in Figure~\ref{fig:branching}, then $F_1$ could be the spanning 3-branching depicted in Figure~\ref{figb:branching}, $F_2$ could be the spanning 2-branching depicted in Figure~\ref{figa:phase2}, and $T$ could be the arborescence in Figure~\ref{figb:phase3}.
In this example, we have~$k_1 = 3$, $N_1 = 25$, $k_2 = 4$, and $N_2 = 30$.

\begin{lemma}\label{lem:lower}
    Let $T$ be the arborescence produced by \Call{MaxLeaves}{$D$}.
    Then \[\ell(T) \ \geq \ \frac{N_1-k_1}6 + \frac{N_2-k_2}2 + 1 \; .\]
\end{lemma}
\begin{proof}
    Let $n$ be the number of vertices of $D$. 
    Let $T_1,\ldots,T_{k_1}$ be the non-trivial arborescences in $F_1$. 
    Note that $\ell(T_j) \geq \frac{1+2|V(T_j)|}3$ because all internal vertices of~$T_j$ have out-degree at least~3.
    Therefore, 
    \begin{align*}
        \ell(F_1) & \ = \ n - N_1 + \sum_{j=1}^{k_1} \ell(T_j)
                    \ \geq \  n - N_1 + \sum_{j=1}^{k_1} \frac{1+2|V(T_j)|}3 \\
                  & \ = \ n - N_1 + \frac{2N_1}3 + \frac{k_1}3 
                    \ = \ n - \frac{N_1-k_1}3 \; .
    \end{align*}

    The number of components in $F_i$ is $n-N_i+k_i$ for $i=1,2$.
    Hence, the number of leaves lost from $F_1$ to $F_2$ is exactly
    \begin{align}\label{eq:leaveslostF1toF2}
      \frac{(n-N_1+k_1) - (n-N_2+k_2)}2 & \ = \ \frac{N_2-k_2}2 - \frac{N_1-k_1}2 \; .
    \end{align}
    Also, the number of leaves lost from $F_2$ to $T$ is exactly $n - N_2 + k_2 - 1 = n - (N_2-k_2) - 1$.
    Thus
    \begin{align*}
        \ell(T) & \ \geq \ n - \frac{N_1-k_1}3 - \left(\frac{N_2-k_2}2 - \frac{N_1-k_1}2\right) - (n-(N_2-k_2)-1) \\
                & \ = \ \frac{N_1-k_1}6 + \frac{N_2-k_2}2 + 1 \; .
    \end{align*}
\end{proof}

For $D$, $F_1$, $F_2$, and $T$ as in Figures~\ref{fig:branching} and~\ref{fig:maxleaves}, we have that $\ell(T) = 18$ and Lemma~\ref{lem:lower} gives as lower bound on $\ell(T)$
$$\frac{N_1-k_1}6 + \frac{N_2-k_2}2 + 1 \ = \ \frac{25-3}6 + \frac{30-4}2 + 1 
                                        \ = \ \frac{11}3+14 = 17.666\ldots$$

Now we are going to present two upper bounds on $\opt(D)$.
The following upper bound holds because the branching $F_2$ could be produced as output of the \Call{Expansion}{} algorithm from Schwartges, Spoerhase, and Wolff~\cite{SchwartgesSW2012}.

\begin{lemma}[Lemma~5~\cite{SchwartgesSW2012}]\label{lem:theirupper}
    It holds that $\opt(D) \leq N_2 - k_2 + 1$.
\end{lemma}

The next lemma is the key for the approximation ratio analysis. 


\begin{lemma}\label{lem:ourupper}
    It holds that $\opt(D) \leq \dfrac{N_1-k_1}2 + \dfrac{N_2-k_2}2 + 1$.
\end{lemma}
\begin{proof}
    We apply on $F_1$ the same definition of witness that Schwartges, Spoerhase, and Wolff~\cite{SchwartgesSW2012} used in their proof of Lemma~\ref{lem:theirupper}. 
    Let $T^*$ be a spanning arborescence of $D$ with the maximum number of leaves.  
    Call $R$ the set of all roots of non-trivial components of $F_1$. 
    Call $L$ the set of leaves of $T^*$ that are isolated vertices of $F_1$. 
    Let $Z := L \cup R \setminus \{r\}$, where $r$ is the root of $D$. 
    See Figure~\ref{fig:opt}.
    The witness of a vertex $z \in Z$ is the closest proper predecessor $q(z)$ of $z$ in $T^*$ which is in a non-trivial component of $F_1$.  
    Note that each witness is an internal vertex of $T^*$.
    These witnesses will not necessarily be pairwise distinct, as in~\cite{SchwartgesSW2012}.
    See Figure~\ref{fig:witnesses}.

\begin{figure}[h]
\centering
\begin{subfigure}{.47\textwidth}
  \includegraphics[width=.98\textwidth,center]{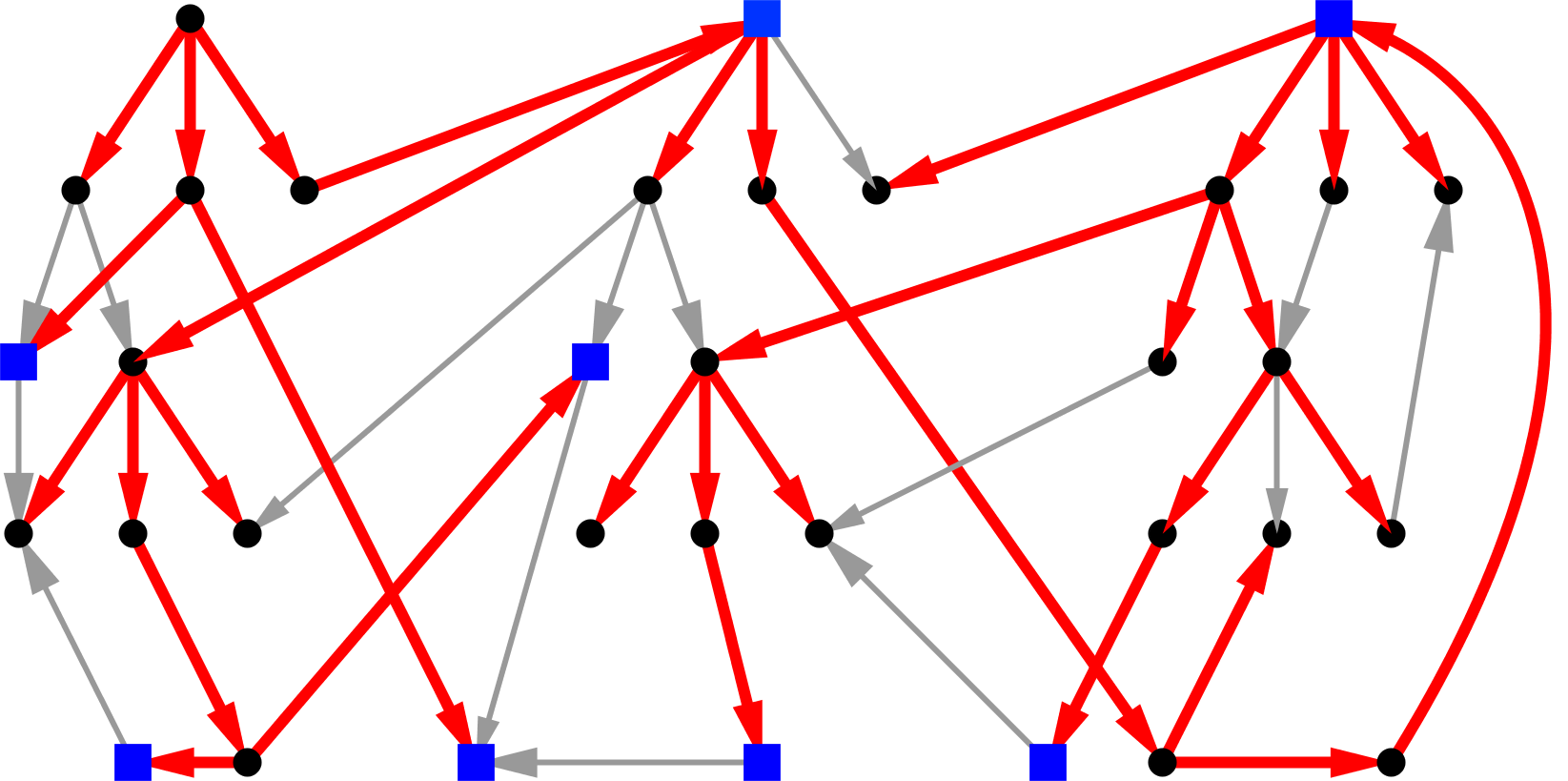}
  \caption{The red and bold arcs show $T^*$, and blue and square vertices show the set~$Z$.}
  \label{fig:opt}
\end{subfigure}
\hspace{5mm}
\begin{subfigure}{.47\textwidth}
  \includegraphics[width=.98\textwidth,center]{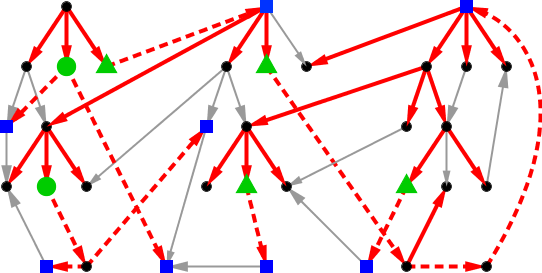}
  \caption{Path between each vertex in $Z$ and its witness in dashed arcs.}
\label{fig:witnesses}
\end{subfigure}
  \caption{The green triangular vertices are the witnesses for a vertex in $Z$ and the green big circles mark two vertices that are the witnesses for two vertices in~$Z$.}
  \label{fig:maxexpand}
\end{figure}


    We will prove that the number $w$ of distinct witnesses is
    \begin{equation}
    \label{eq:main}
        w \ \geq \ |Z| - \frac{N_2-k_2}{2} + \frac{N_1-k_1}{2} 
            \ = \ k_1 - 1 + |L| - \frac{N_2-k_2}{2} + \frac{N_1-k_1}{2} \; .
    \end{equation}
    From this, because each witness lies in a non-trivial component of~$F_1$ and is internal in $T^*$, we deduce that
    \begin{align*}
        \opt(D) & \ \leq \ N_1 - \left(k_1 - 1 + |L| - \frac{N_2-k_2}{2} + \frac{N_1-k_1}{2}\right) + |L| \\
                & \ = \ N_1 - k_1 + \frac{N_2-k_2}{2} - \frac{N_1-k_1}{2} + 1 \ = \ \frac{N_1 - k_1}2 + \frac{N_2-k_2}{2} + 1 \; .
    \end{align*}
    It remains to prove~\eqref{eq:main}.

    For a witness $s$, let $Z_s := \{ z \in Z: q(z) = s\}$.  
    Let $T^*_s$ be the subarborescence of~$T^*$ induced by the union of all paths in $T^*$ from $s$ to each vertex in~$Z_s$.  
    The number of such arborescences $T^*_s$ is exactly $w$.  
    Note that the only internal vertex of~$T^*_s$ that is in a non-trivial component of~$F_1$ is its root $s$, which is necessarily a leaf of~$F_1$ (because there is no arc from an internal vertex of~$F_1$ to vertices out of $F_1$).
    Thus the maximum out-degree in~$T^*_s$ is at most two. 

    First let us argue that no $z$ in $Z_s$ is a predecessor in $T^*_s$ of another $z'$ in~$Z_s$. 
    Suppose by contradiction that $z$ is in the path from $s$ to $z'$.  
    Then $z$ is not a leaf of $T^*$ and therefore $z$ is in $R$, and thus is in a non-trivial component of~$F_1$.  
    This leads to a contradiction because~$z$, and not~$s$, would be the witness for~$z'$.
    Therefore $T^*_s$ has exactly $|Z_s|$ leaves.  

    Let $B$ be the set of vertices $v$ such that $e_v \in M$, where $M$ is the maximum matching in the multigraph $G$ computed during the execution of \mbox{\Call{MaxExpand}{$D$,~$F_1$}.}  
    Observe that $|M|$ is exactly the number of leaves lost from branching~$F_1$ to $F_2$, given by~\eqref{eq:leaveslostF1toF2}, so
    \begin{equation}\label{eq:B}
        |M| \ = \ |B| \ = \ \frac{N_2-k_2}2 - \frac{N_1-k_1}2 \; .
    \end{equation}
    Now let us argue that the vertices with out-degree two in $T^*_s$ are all in the set~$\Candidates$, defined 
    in Algorithm~\ref{alg:3}.
    Let~$v$ be one such vertex.
    Either~$v$ is an isolated vertex or~$v$ is a leaf of a non-trivial component of~$F_1$.  
    Therefore~$d_{F_1}^+(v) = 0$.  
    As the two children of $v$ in $T^*_s$ have in-degree~0 in~$F_1$, both are in~$U_v$. 
    Hence~$v \in \Candidates$. 

    Let $C_s$ be the set of vertices of $\Candidates$ with out-degree two in~$T^*_s$ and $C = \cup_s C_s$.  
    Then the number of leaves in~$T^*_s$ is $|Z_s| = |C_s|+1$. 
    The set of internal vertices of~$T^*_s$ and of~$T^*_{s'}$ are disjoint for distinct witnesses~$s$ and~$s'$.  
    Thus the sets~$C_s$ and $C_{s'}$ are disjoint.  
    Let $M_C$ be the set of edges of $G$ corresponding to the vertices in $C$.  
    Note that $M_C$ is a matching, so~$|C| = |M_C| \leq |M| = |B|$. 
    Hence
    \[ |Z| \ = \ \sum_s |Z_s| \ = \ \sum_s (|C_s|+1) \ = \ |C| + w \ \leq \ |B| + w \; . \]
    Therefore $w \geq |Z| - |B| = k_1 - 1 + |L| - (\frac{N_2-k_2}2 - \frac{N_1-k_1}2)$, as in~\eqref{eq:main}.
\end{proof}

Continuing with our example, if $D$ is the dag depicted in Figure~\ref{fig:branching}, 
then Lemma~\ref{lem:theirupper} implies that $\opt(D) \leq 27$, 
while Lemma~\ref{lem:ourupper} implies that ${\opt(D) \leq 25}$. 


\begin{theorem}\label{thm:32}
    Algorithm \Call{MaxLeaves}{} is a $\frac32$-approximation for the \maxleaves\ on rooted directed acyclic graphs.
\end{theorem}
\begin{proof}
    For a rooted \Dag\ $D$, let $T$ be the output of \Call{MaxLeaves}{$D$}. 
    Then
    \begin{align}
        \ell(T) & \ \geq \ \frac{N_1{-}k_1}6 + \frac{N_2{-}k_2}2 + 1 \label{eq1} \\
                & \ = \ \frac{N_1{-}k_1}6 + \frac{N_2{-}k_2}6 + \frac{N_2{-}k_2}3 + 1 \nonumber \\
                & \ \geq \ \frac{\opt(D){-}1}3 + \frac{\opt(D){-}1}3 + 1 \label{eq2} \\
                & \ > \ 2\,\frac{\opt(D)}3 \; , \nonumber
    \end{align}
    where~\eqref{eq1} holds by Lemma~\ref{lem:lower} and~\eqref{eq2} holds by Lemmas~\ref{lem:theirupper} and~\ref{lem:ourupper}.
\end{proof}

The bound given in Theorem~\ref{thm:32} is tight.  
Indeed, an example similar to the one by Schwartges, Spoerhase, and Wolff~\cite{SchwartgesSW2012} for their algorithm proves that algorithm \Call{MaxLeaves}{} can achieve ratios arbitrarily close to $3/2$. 
See Figure~\ref{fig:tight}.

\begin{figure}[ht]
    \centerline{\includegraphics[height=4cm]{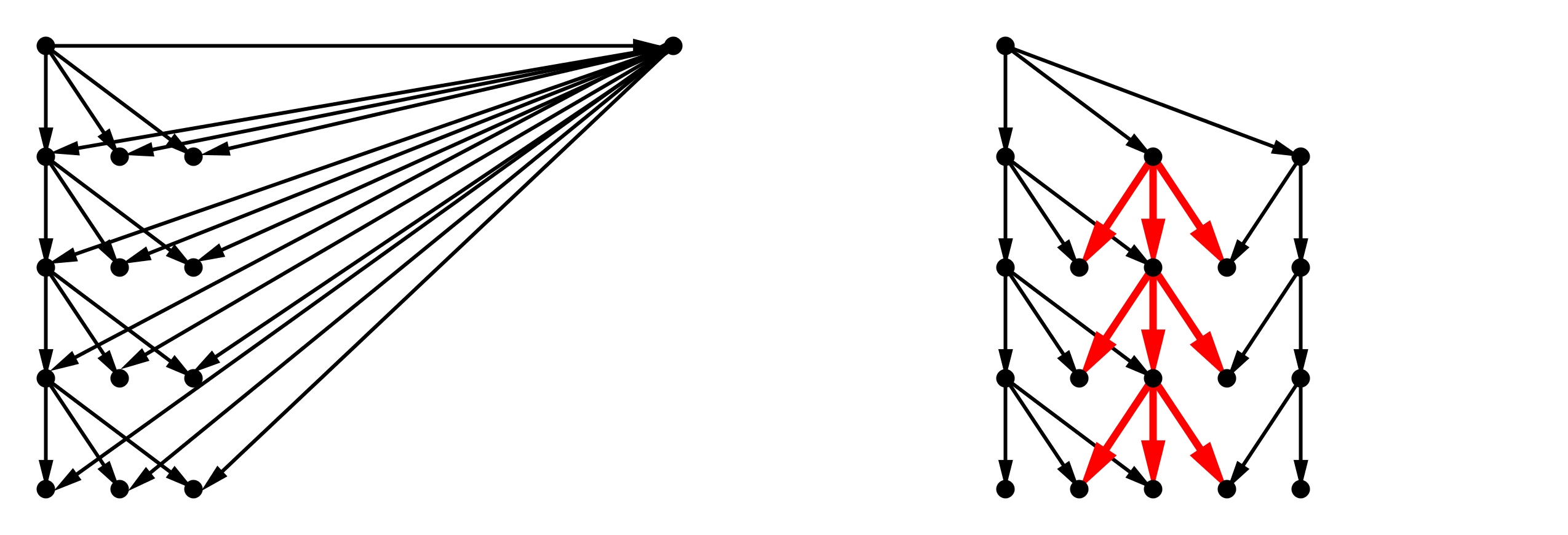}}
    \caption{In the left, a rooted dag with $n=3k+2$ vertices, for $k=4$, where {\sc MaxLeaves} can produce an arborescence with $2k+2$ leaves and $\opt = 3k$ leaves.
    In the right, an example where the 3-expansions in red and bold damage all possible 2-expansions that exist in the optimal arborecence.}
    \label{fig:tight}
\end{figure}

\section{Using approximations for 3-dimensional matching}
\label{sec:3dmatching}

The problem known as \DM, from 3-dimensional matching, consists in the following. 
Given a finite set $U$ and a collection $\calS$ of subsets of~$U$ with three elements each, find a collection $\calS' \subseteq \calS$ of pairwise disjoint sets with as many sets as possible.
The name of the problem comes from the fact that one can define a 3-regular hypergraph on the vertex set~$U$ whose edges are the sets in~$\calS$, and the collection~$\calS'$ is a maximum matching in such hypergraph.  

This problem is \NP-hard~\cite{GareyJ1979}, and there exists a $\frac43$-approximation for it~\cite{Cygan2013,FurerY2014} as well as a~$(2+\eps)$-approximation for any $\eps>0$, for its weighted version~\cite{ArkinH1998,Berman2000}. 

The strategy in \Call{MaxLeaves}{} can be generalized by using an approximation algorithm for \DM.
One possibility is, for a rooted dag $D$, to call \Call{GreedyExpand}{$D$, $4$, $F_0$} with the empty spanning branching $F_0$, obtaining~$F_1$, then to use the $\frac43$-approximation for \DM\ to expand~$F_1$ with a good set of $3$-expansions, resulting in a branching $F_2$.
Then we can proceed as in \Call{MaxLeaves}{}, calling \Call{MaxExpand}{$D$, $F_2$} to obtain~$F_3$, and \Call{GreedyExpand}{$D$, $1$, $F_3$} to obtain the final arborescence $T$.
This idea however will not give an approximation better than~$\frac32$.
See the example to the right in Figure~\ref{fig:tight}.
It shows that, to achieve a ratio better than $\frac32$, the choice of a good 3-dimensional matching has to somehow take into consideration the 2-expansions.

The following weighted variant of the previous idea takes into account 2-expansions and 3-expansions simultaneously.
This variant does not give an improvement, but it can have implications for \maxleaves\ on rooted dags if a better approximation for the weighted variant of \DM\ is designed.  
So we describe it ahead.

\newcommand{\FF}{F^{w}}
\newcommand{\NN}{N^{w}}
\newcommand{\kk}{k^{w}}

The second possibility we investigated makes use of weights in the following way.
We start by calling \Call{GreedyExpand}{\mbox{$D$, $4$, $\FF_0$}} with the empty spanning branching~$\FF_0$, obtaining~$\FF_1$.  
After that, we create an instance of the weighted \DM\ where feasible $3$-expansions turn into sets of weight two and feasible $2$-expansions turn into sets of weight one.
We then use an approximation for the weighted \DM\ to obtain a branching~$\FF_3$ from~$\FF_1$. 
We finish by calling \Call{GreedyExpand}{$D$, $1$, $\FF_3$} to obtain an arborescence.
This algorithm, named \Call{Maxleaves-W3DM}{}, is formalized in Algorithm~\ref{alg:maxleaves-w3dm}.
It makes use of a procedure named \Call{ApproxMaxWeighted3DMatching}{$\calS$, $w$}, that receives a collection $\calS$ of sets of size two or three, and a weight function~$w$ defined on $\calS$.  
The procedure returns a subcollection of $\calS$ that consists of pairwise disjoint sets.
The weight of this subcollection depends on the approximation guarantee of the procedure.

To convert $\calS$ into a collection of sets of size three, as in an instance of \DM, one can add a new element to each set of size two, keeping the same weights.  
For simplicity, we refrain from doing this, and in the proof of Lemma~\ref{lemma:w3dm_upper_opt} ahead, we abuse notation and refer to subcollections of $\calS$ of pairwise disjoint sets as 3D-matchings of $\calS$.

Next lemma corresponds to Lemma~\ref{lem:lower}, and their proofs are quite similar.
Let~$\FF_1$, $\FF_2$, $\FF_3$, and $T$ be the branchings produced during the call \Call{MaxLeaves-W3DM}{$D$}.
For~$i=1,2,3$, let~$\kk_i$ be the number of non-trivial components of~$\FF_i$ and~$\NN_i$ be the number of vertices in such components.


\begin{algorithm}
\begin{algorithmic}
\Require{rooted acyclic directed graph $D$}
\Ensure{spanning arborescence of $D$}

\State let $\FF_0$ be the spanning branching with no arcs
\State $\FF_1 \gets$ \Call{GreedyExpand}{$D$, $4$, $\FF_0$}
\For{each $v \in V(D)$ such that $d_{\FF_1}^+(v)=0$}
    \State $U_v \gets \{vu \in A(D) : d_{\FF_1}^-(u)=0\}$
\EndFor
\State $\Candidates \gets \{v \in V(D) : d_{\FF_1}^+(v)=0 \mbox{ and } 2 \leq |U_v| \leq 3\}$
\State $\calS \gets \{U_v : v \in \Candidates\}$
\State $\calS \gets \calS \cup \{\{a,b\},\{a,c\},\{b,c\} : v \in \Candidates \mbox{ and } |U_v| = \{a,b,c\}\}$
\State $w(U) \gets |U|-1$ for each $U \in \calS$
\State $M \gets$ \Call{ApproxMaxWeighted3DMatching}{$\calS$, $w$}
\State $\FF_2 \gets \FF_1$
\For{each $U \in M$ such that $|U| = 3$}
    \State $\FF_2 \gets \FF_2 + U$
\EndFor
\State $\FF_3 \gets \FF_2$
\For{each $U \in M$ such that $|U| = 2$}
    \State $\FF_3 \gets \FF_3 + U$
\EndFor
\State $T \gets$ \Call{GreedyExpand}{$D$, $1$, $\FF_3$}
\State \Return $T$
\end{algorithmic}
\caption{\textsc{MaxLeaves-W3DM}($D$)}
\label{alg:maxleaves-w3dm}
\end{algorithm}

\newpage

\begin{lemma}
\label{lem:lower_w3dm}
    Let $T$ be the arborescence produced by \Call{MaxLeaves-W3DM}{$D$}.
    Then \[\ell(T) \ \geq \ \frac{\NN_1-\kk_1}{12} + \frac{\NN_2-\kk_2}{6} + \frac{\NN_3-\kk_3}{2} + 1 \; .\]
\end{lemma}
\begin{proof}
    Let $n$ be the number of vertices of $D$. 
    Let $T_1,\ldots,T_{\kk_1}$ be the non-trivial arborescences in~$\FF_1$. 
    Note that $\ell(T_j) \geq \frac{1+3|V(T_j)|}4$ because all internal vertices of~$T_j$ have out-degree at least~4.
    Therefore, 
    \begin{align*}
        \ell(\FF_1) & \ = \ n - \NN_1 + \sum_{j=1}^{\kk_1} \ell(T_j)
                    \ \geq \  n - \NN_1 + \sum_{j=1}^{\kk_1} \frac{1+3|V(T_j)|}4 \\
                  & \ = \ n - \NN_1 + \frac{3\NN_1}4 + \frac{\kk_1}4 
                    \ = \ n - \frac{\NN_1-\kk_1}4 \; .
    \end{align*}

    The number of components in $\FF_i$ is $n-\NN_i+\kk_i$ for $i=1,2,3$.
    Hence, the number of leaves lost from $\FF_1$ to $\FF_2$ is exactly 
    \[\frac{(n-\NN_1+\kk_1) - (n-\NN_2+\kk_2)}3 \ = \ \frac{\NN_2-\kk_2}3 - \frac{\NN_1-\kk_1}3 \; .\] 
    Similarly, the number of leaves lost from $\FF_2$ to $\FF_3$ is exactly
    \[\frac{(n-\NN_2+\kk_2) - (n-\NN_3+\kk_3)}2 \ = \ \frac{\NN_3-\kk_3}2 - \frac{\NN_2-\kk_2}2 \; .\] 
    Also, the number of leaves lost from $\FF_3$ to $T$ is exactly $n - \NN_3 + \kk_3 - 1 = n - (\NN_3-\kk_3) - 1$.
    Thus
    \begin{align*}
        \ell(T) & \ \geq \ n - \frac{\NN_1-\kk_1}4 - \left(\frac{\NN_2-\kk_2}3 - \frac{\NN_1-\kk_1}3\right) \\
                & \phantom{\ \geq \ n \ } - \left(\frac{\NN_3-\kk_3}2 - \frac{\NN_2-\kk_2}2\right) - (n-(\NN_3-\kk_3)-1) \\
                & \ = \ \frac{1}{12} (\NN_1-\kk_1) + \frac16 (\NN_2-\kk_2) + \frac12 (\NN_3-\kk_3) + 1 \; .
    \end{align*}
\end{proof}

Now we present an upper bound on $\opt(D)$ that relates to the lower bound presented in Lemma~\ref{lem:lower_w3dm}.
Its proof follows the lines of the proof of Lemma~\ref{lem:ourupper}.

\begin{lemma}\label{lemma:w3dm_upper_opt}
  If algorithm \Call{ApproxMaxWeighted3DMatching}{} is an $\alpha$-approximation for the weighted \DM, then 
  $$\opt(D) \ \leq \ \frac{3-2\alpha}{3}(\NN_1-\kk_1) + \frac{\alpha}{6}(\NN_2-\kk_2) + \frac{\alpha}{2}(\NN_3-\kk_3) + 1.$$ 
\end{lemma}
\begin{proof}
    Let $T^*$ be a spanning arborescence of $D$ with the maximum number of leaves.  
    Call $R$ the set of all roots of non-trivial components of $\FF_1$. 
    Call $L$ the set of leaves of $T^*$ that are isolated vertices of $\FF_1$. 
    Let $Z := L \cup R \setminus \{r\}$, where $r$ is the root of $D$. 
    The witness of a vertex $z \in Z$ is the closest proper predecessor $q(z)$ of $z$ in $T^*$ which is in a non-trivial component of $\FF_1$.
    Note that each witness is an internal vertex of $T^*$.

    We will prove that the number $\psi$ of distinct witnesses is
    \begin{align}
    \label{eq:witness_number_w3dm}
        \psi & \geq |Z| - 2\alpha\left(\frac{\NN_2-\kk_2}3 - \frac{\NN_1-\kk_1}3\right) - \alpha\left(\frac{\NN_3-\kk_3}2 - \frac{\NN_2-\kk_2}2\right) \\ 
             & = |Z| + 2\alpha\,\frac{\NN_1-\kk_1}3 - \alpha\,\frac{\NN_2-\kk_2}6 - \alpha\,\frac{\NN_3-\kk_3}2 \; \nonumber. 
    \end{align}
    Because $|Z| = \kk_1 - 1 + |L|$ and each witness lies in a non-trivial component of~$\FF_1$ and is internal in $T^*$, we deduce that
    \begin{align*}
        \opt(D) & \ \leq \ \NN_1 - \psi + |L| \\
                & \ \leq \ \NN_1 - |Z| - 2\alpha\,\frac{\NN_1-\kk_1}3 + \alpha\,\frac{\NN_2-\kk_2}6 + \alpha\,\frac{\NN_3-\kk_3}2 + |L| \\
                & \ = \ \NN_1 - \kk_1 - 2\alpha\,\frac{\NN_1-\kk_1}3 + \alpha\,\frac{\NN_2-\kk_2}6 + \alpha\,\frac{\NN_3-\kk_3}2 + 1 \\
                & \ = \ \frac{3-2\alpha}{3}(\NN_1-\kk_1) + \frac{\alpha}{6}(\NN_2-\kk_2) + \frac{\alpha}{2}(\NN_3-\kk_3) + 1 \; .
    \end{align*}
    It remains to prove~\eqref{eq:witness_number_w3dm}.
    The proof follows closely to that of~\eqref{eq:main}.

    For a witness $s$, let $Z_s := \{ z \in Z: q(z) = s\}$ and let $T^*_s$ be the subarborescence of~$T^*$ induced by the union of all paths in $T^*$ from $s$ to each vertex in~$Z_s$.  
    The number of such arborescences $T^*_s$ is exactly $\psi$.  
    The only internal vertex of~$T^*_s$ that is in a non-trivial component of~$\FF_1$ is its root $s$, which is necessarily a leaf of~$\FF_1$. 
    So the maximum out-degree in~$T^*_s$ is at most three. 

    Again, no $z \in Z_s$ is a predecessor in $T^*_s$ of another $z' \in Z_s$. 
    Indeed, suppose by contradiction that $z$ is in the path from $s$ to $z'$.  
    Then $z$ is not a leaf of~$T^*$, and is in $R$, thus being in a non-trivial component of~$\FF_1$, which is a contradiction, because~$z$, and not~$s$, would be the witness for~$z'$.
    Hence $T^*_s$ has exactly $|Z_s|$ leaves.  

    When $|U_v| = 3$, the algorithm includes in $\calS$ also the subsets of $U_v$ with two elements. 
    At most one among $U_v$ and these subsets is included in $M$, where~$M$ is the output of \Call{ApproxMaxWeighted3DMatching}{$\calS$, $w$}, computed during the execution of \Call{MaxLeaves-W3DM}{$D$}.  
    Thus, in what follows, we abuse notation and refer to each such subset of $U_v$ also as $U_v$.
    Let~$B_i$ be the set of vertices $v$ such that $U_v \in M$ and $|U_v| = i$, for $i = 2,3$. 
    Note that~$|B_3|$ is exactly the number of leaves lost from branching~$\FF_1$ to $\FF_2$, so 
    \begin{equation}
    \label{eq:w3dm_b3}
        |B_3| \ = \ \frac{\NN_2-\kk_2}3 - \frac{\NN_1-\kk_1}3 \; .
    \end{equation}
    Also, $|B_2|$ is exactly the number of leaves lost from branching~$\FF_2$ to $\FF_3$, so
    \begin{equation}
    \label{eq:w3dm_b2}
        |B_2| \ = \ \frac{\NN_3-\kk_3}2 - \frac{\NN_2-\kk_2}2 \; .
    \end{equation}
    Finally, $|M| = |B_3| + |B_2|$ and $w(M) = 2|B_3| + |B_2|$.

    Vertices with out-degree two and three in $T^*_s$ are all in the set~$\Candidates$.
    Indeed, let~$v$ be one such vertex.
    Either~$v$ is an isolated vertex or~$v$ is a leaf of a non-trivial component of~$\FF_1$.  
    So~$d_{\FF_1}^+(v) = 0$.  
    As the children of $v$ in $T^*_s$ have in-degree 0 in~$\FF_1$, they are all in~$U_v$. 
    Hence~$v \in \Candidates$. 

    For $i = 2,3$, let $C^i_s$ be the set of vertices of $\Candidates$ with out-degree~$i$ in~$T^*_s$, and let~$C = \cup_s C^i_s$.  
    The number of leaves in~$T^*_s$ is ${|Z_s| = 2|C^3_s|+|C^2_s|+1}$. 
    The set of internal vertices of~$T^*_s$ and of~$T^*_{s'}$ are disjoint for distinct witnesses~$s$ and~$s'$.  
    Thus the sets~$C^i_s$ and $C^i_{s'}$ are disjoint.  
    Let $M_C$ be the subset of $\calS$ corresponding to the vertices in $C$.
    Note that $M_C$ is a 3D-matching of~$\calS$, so~$w(M_C) = 2|C^3_s| + |C^2_s| \leq w(M^*) \leq \alpha\,w(M)$, 
    where $M^*$ is a maximum weighted 3D-matching of $\calS$.
    Hence
    \begin{align*}
        |Z| \ &= \ \sum_s |Z_s| \ = \ \sum_s (2|C^3_s|+|C^2_s|+1) \ = \ w(M_C) + \psi \\
              &\leq \ \alpha\,w(M) + \psi \ = \ 2\alpha|B_3| + \alpha|B_2| + \psi \; . 
    \end{align*}
    Therefore, 
    \begin{align*}
      \psi & \ \geq \ |Z| - 2\alpha|B_3| - \alpha|B_2| \\
           & \ = \ |Z| - 2\alpha\left(\frac{\NN_2-\kk_2}3 - \frac{\NN_1-\kk_1}3\right) - \alpha\left(\frac{\NN_3-\kk_3}2 - \frac{\NN_2-\kk_2}2\right)\,,
    \end{align*}
    which completes the proof of~\eqref{eq:witness_number_w3dm}.
\end{proof}

\begin{theorem}
\label{thm:weighted_approx}
  If algorithm \Call{ApproxMaxWeighted3DMatching}{} is an $\alpha$-approximation for the weighted \DM, then algorithm \Call{MaxLeaves-W3DM}{} is a $\max\{\frac43,\alpha\}$-approximation for the \maxleaves\ on rooted directed acyclic graphs.
\end{theorem}
\begin{proof}
  For a rooted \Dag\ $D$, let $T$ be the output of \Call{MaxLeaves-W3DM}{$D$}. 
  First, suppose $\alpha \geq \frac43$. 
  In this case, $\frac{3-2\alpha}{3} \leq \frac{\alpha}{12}$ and, by Lemmas~\ref{lem:lower_w3dm} and~\ref{lemma:w3dm_upper_opt}, 
  \begin{align*}
    \opt(D) & \ \leq \ \frac{3-2\alpha}{3}(\NN_1-\kk_1) + \frac{\alpha}{6}(\NN_2-\kk_2) + \frac{\alpha}{2}(\NN_3-\kk_3) + 1 \\
            & \ \leq \ \frac{\alpha}{12} (\NN_1-\kk_1) + \frac{\alpha}6 (\NN_2-\kk_2) + \frac{\alpha}2 (\NN_3-\kk_3) + \alpha \\
            & \ \leq \ \alpha\,\ell(T)\; . 
  \end{align*}

  Now, suppose $\alpha < \frac43$, and let $\beta = \frac43 - \alpha$.
  By Lemma~\ref{lemma:w3dm_upper_opt},
  \begin{align}
    \opt(D) & \ \leq \ \frac{3-2\alpha}{3}(\NN_1-\kk_1) + \frac{\alpha}{6}(\NN_2-\kk_2) + \frac{\alpha}{2}(\NN_3-\kk_3) + 1 \nonumber \\
            & \ \leq \ \left(\frac19+\frac23\,\beta\right)(\NN_1-\kk_1) + \left(\frac29-\frac16\,\beta\right)(\NN_2-\kk_2) \nonumber \\
            & \phantom{\ \leq \ } + \left(\frac23-\frac12\,\beta\right)(\NN_3-\kk_3) + 1 \nonumber \\
            & \ = \ \frac19(\NN_1-\kk_1) + \frac29(\NN_2-\kk_2) + \frac23(\NN_3-\kk_3) + \frac43 \nonumber \\ 
            & \phantom{\ \leq \ } + \frac23\,\beta(\NN_1-\kk_1) - \frac16\,\beta(\NN_2-\kk_2) - \frac12\,\beta(\NN_3-\kk_3) - \frac13 \nonumber \\
            & \ \leq \ \frac43\,\ell(T) + \frac23\,\beta\left((\NN_1{-}\kk_1) - \frac14(\NN_2{-}\kk_2) - \frac34(\NN_3{-}\kk_3)\right) - \frac13\label{eq:lT}\\
            & \ \leq \ \frac43\,\ell(T)\;, \label{eq:negative} 
  \end{align}
  where~\eqref{eq:lT} holds by Lemma~\ref{lem:lower_w3dm} and~\eqref{eq:negative} holds because the number of components in~$\FF_1$, $\FF_2$, and $\FF_3$ is so that $n-\NN_1+\kk_1 \geq n-\NN_2+\kk_2 \geq n-\NN_3+\kk_3$, and this implies that $\NN_1-\kk_1 \leq \NN_2-\kk_2 \leq \NN_3-\kk_3$, and therefore $\NN_1-\kk_1 \leq \frac14(\NN_2-\kk_2) + \frac34(\NN_3-\kk_3)$. 
\end{proof}

Note that the weighted instance of \DM\ we used has only weights~$1$ and~$2$. 
So a good approximation even for this more restricted weighted version of \DM\ would be of interest. 

At the moment, because the best approximation for the weighted \DM\ has ratio greater than~$2$, this does not provide an improvement on the previously best known ratio for \maxleaves.
Now, only a ratio better than $3/2$ for the weighted \DM\ would provide an improvement.

\section{Inapproximability of the vertex-weighted version}
\label{sec:inapproximability}

A vertex-weighted generalization of the maximum leaf spanning tree (the undirected version of our problem) was considered in the literature.
In such generalization, one is given a connected vertex-weighted graph and the goal is to find a spanning tree whose sum of leaf weights is maximum.  

Jansen~\cite{Jansen2012} proved that, unless $\PP = \NP$, this version of the problem does not admit a polynomial-time ratio $\Oh(n^{\frac12-\eps})$ or a $\Oh(\opt^{\frac13-\eps})$-approximation for any $\eps>0$, where $n$ is the number of vertices of the given graph. 
His reduction is from the \IS\ problem.  
A straightforward modification of his reduction shows the same inapproximability results for the vertex-weighted version of \maxleaves\ on rooted dags. 
Next we describe his reduction adapted to produce rooted dags with binary weights. 

The \IS\ problem consists of the following: given a graph~$G$, find an independent set in~$G$ with as many vertices as possible.

Let $G$ be an instance of the \IS\ problem. 
Let $D$ be the rooted dag that has as vertices the vertices of~$G$, a new vertex~$r$ as its root, and a vertex~$e$ for each edge $e$ of~$G$.  
There is an arc from $r$ to each vertex of~$G$ in~$D$.
For each edge~$e=uv$ of $G$, there is an arc from~$u$ to~$e$ and an arc from~$v$ to~$e$ in~$D$. 
So, if~$G$ has~$n$ vertices and~$m$ edges, $D$ has $n+m+1$ vertices and $n+2m$ arcs.
See Figure~\ref{fig:IS}.
Note that~$D$ is~$r$-rooted and acyclic and that, in any spanning arborescence in $D$, the vertices corresponding to edges of $G$ are leaves, because they have out-degree~0 in $D$.
Because the complement of an independent set is an edge cover, the following holds.

\begin{figure}[htb]
    \centerline{\includegraphics[height=2.5cm]{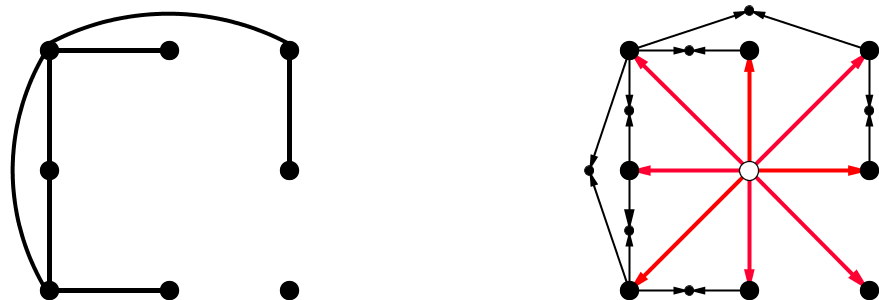}}
    \caption{An instance of the \IS\ problem and the corresponding rooted dag. The white vertex is the root of the dag.}
    \label{fig:IS}
\end{figure}

\begin{lemma}
\label{lem:ISversusWLSA}
    Set $S$ is an independent set in $G$ if and only if there is a spanning arborescence in~$D$ that has $S \cup E(G)$ as leaves.
\end{lemma}
\begin{proof}
  Let $S$ be an independent set in $G$.
  Initialize an arborescence $T$ by taking the root of $D$ and adding arcs to every vertex of $V(G)$.
  Since $V(G) \setminus S$ is an edge cover in $G$, we can augment~$T$ to a spanning arborescence by connecting vertices of $D$ in $V(G) \setminus S$ to vertices of $D$ in $E(G)$.
  Thus all vertices in $S \cup E(G)$ are leaves of $T$.
  
  Now let $T$ be a spanning arborescence of $D$ that has $S \cup E(G)$ as leaves, with $S \subseteq V(G)$.  
  (Recall that any spanning arborescence of $D$ has all vertices in $E(G)$ as leaves.)
  Thus, for $u,v \in S$, if~$uv \in E(G)$, then $e = uv \in V(D)$ would be 
  an isolated vertex in $T$ because both of its in-neighbors~$u$ and $v$ are leaves in $T$, a contradiction.
  Therefore $S$ is an independent set of $G$.
\end{proof}

For any $\eps > 0$, there is no polinomial-time $\Oh(n^{1-\eps})$-approximation for \IS\ unless $\PP = \NP$~\cite{Hastad1999}, where~$n$ is the number of vertices of the given graph~$G$.
Using this and Lemma~\ref{lem:ISversusWLSA}, we derive the following.

\begin{theorem}
\label{thm:inapprox}
    The vertex-weighted \maxleaves\ on directed acyclic graphs with binary weights and maximum in-degree~2
    does not have a polinomial-time $\Oh(n^{1-\eps})$-approximation for any~$\eps > 0$, unless $\PP = \NP$,
    where~$n$ is the number of weight-one vertices of the given directed graph. 
\end{theorem}
\begin{proof}
  We will describe an approximation-preserving reduction from \IS\ to the vertex-weighted \maxleaves.

  Let $G$ be an instance of the \IS\ problem. 
  Let $D$ be the rooted dag defined from $G$ as before Lemma~\ref{lem:ISversusWLSA}.
  Assign weights to the vertices of~$D$ as follows: vertices of~$G$ have weight~$1$ while vertices corresponding to edges of~$G$ have weight~$0$.
  The root~$r$ may have an arbitrary weight, because it will never be a leaf in a spanning arborescence of~$D$. 
  Note that the weights are binary and that the maximum in-degree in $D$ is~2.

  Let $T^*$ be a maximum leaf weighted arborescence of $D$ and $S^*$ be a maximum independent set in $G$.  Note that $w(T^*) = |S^*|$ by Lemma~\ref{lem:ISversusWLSA}.
  
  Suppose that $A$ is an $\Oh(n^{1-\eps})$-approximation, for some $\eps>0$, for the vertex-weighted \maxleaves\ on dags with binary weights and maximum in-degree~2, where~$n$ is the number of weight-one vertices of the given directed graph.
  Let $T$ be the spanning arborescence of~$D$ obtained from applying $A$ to $D$ with weights $w$.  
  Then, for some constant~${c > 0}$, we have that~$w(T^*) \leq c\,n^{1-\eps}\,w(T)$, where $n = |V(G)|$.
  Let $S$ be the set of~$w(T)$ leaves of~$T$ in $V(G)$ which, by Lemma~\ref{lem:ISversusWLSA}, form an independent set in~$G$.
  Hence,
  \begin{equation}
    |S| \ = \ w(T) \ \geq \ \frac{w(T^*)}{c\,n^{1-\eps}} = \frac{|S^*|}{c\,n^{1-\eps}} \; .
  \end{equation}
  This would then be an $\Oh(n^{1-\eps})$-approximation for the \IS{}, which exists only if~$\PP=\NP$ by~\cite{Hastad1999}.
\end{proof}

To avoid using weight zero, a similar result can be obtained by assigning weights~$m$ and~$1$ instead of $1$ and~$0$, respectively, where $m$ is the number of edges in $G$.
For this case, Lemma~\ref{lem:ISversusWLSA} implies that an independent set of size $t$ in $G$ corresponds to a spanning arborescence of leaf weight~$(t+1)m$, and a similar inapproximability result holds, as Jansen~\cite{Jansen2012} proved for the undirected version.

\section{Future directions}
\label{sec:remarks}

Improving on the 92-approximation for the general directed case would be very interesting.
A major difficulty is that greedy strategies do not apply so easily, because not every branching can be extended to a spanning branching in an arbitrary rooted digraph.
The strategy used by Daligault and Thomass\'{e}~\cite{DaligaultT2009} consists of a series of reductions, and some of them end up with a dag.
It is tempting to try to use an approximation for dags within their algorithm to achieve an improved ratio, however we did not succeed in doing that so far. 

Directed acyclic graphs have directed tree width zero~\cite{JohnsonRST2001}.
Is it possible to extend our approximation or any greedy algorithm for \maxleaves\ to address directed graphs with bounded directed tree width?

It is natural to wonder if there is a way to optimize one of the expansions used in Solis-Oba's algorithm to achieve a better approximation ratio for the undirected case.
Also, for the undirected case, there are better approximations for cubic graphs~\cite{BonsmaZ2011,CorreaFMW2008}.
Maybe one can obtain better bounds on the approximation ratio for dags with out-degree bounded by three or two.

%
%
%
%
%

\bibliographystyle{plain} 
\bibliography{maxleaves}

\end{document}